\documentclass[reprint,superscriptaddress,twocolumn]{revtex4}
\usepackage{amsthm,amssymb,amsmath,graphicx,epsfig,color,verbatim,enumerate}
\usepackage{epstopdf,rotating}

\newcommand{\ket}[1]{\left|#1\right\rangle}

\newcommand{\bracket}[2]{\left\langle #1|#2\right\rangle}
\newcommand\defn[1]{\textsl{#1}}

\newcommand\ketbra[1]{|#1\rangle\langle#1|}
\newcommand\cH{{\mathcal H}}
\newcommand\cM{{\mathcal M}}
\newcommand\cN{{\mathcal N}}

\newcommand\cG{{\mathcal G}}

\newcommand\cB{{\mathcal B}}

\newtheoremstyle{dotless}{}{}{\itshape}{}{\bfseries}{}{ }{}
\theoremstyle{dotless}

\newtheorem{theorem}{Theorem}

\newtheorem{definition}{Definition}

\begin{document}
\title{Alignment of reference frames and an operational interpretation for the
$G$-asymmetry}
\author{Michael Skotiniotis}
\affiliation{Institute
for Quantum Information Science, University of Calgary,2500 University
Drive NW,Calgary AB, T2l 1N1, Canada}
\author{Gilad Gour} 
\affiliation{Institute
for Quantum Information Science, University of Calgary,2500 University
Drive NW,Calgary AB, T2l 1N1, Canada}
\affiliation{Department of Mathematics and Statistics, University of Calgary,
2500 University Drive NW
Calgary, Alberta, T2N 1N4, Canada}
\date{\today}
\begin{abstract}
We determine the quantum states and measurements that optimize the accessible information in a reference frame alignment protocol associated with the groups $U(1)$, corresponding to a phase reference, and $\mathbb{Z}_M$, the cyclic group of $M$ elements.  Our result provides an operational interpretation for the $G$-asymmetry which is information-theoretic and which was thus far lacking. In particular, we show that in the limit of many copies of the bounded-size quantum reference frame, the accessible information approaches the Holevo bound.  This implies that the rate of alignment of reference frames, measured by the (linearized) accessible information per system, is equal to the regularized, linearized $G$-asymmetry. The latter quantity is equal to the number variance in the case where $G=U(1)$. 
Quite surprisingly, for the case where $G=\mathbb{Z}_{M}$ and $M\geq 4$, it is equal to a quantity that is not additive in general, but instead can be superadditive under tensor product of two distinct bounded-size reference frames. 
This remarkable phenomenon is purely quantum and has no classical analog.
\end{abstract}
\maketitle

\section{\label{sec:1}Introduction}
States and transformations of systems are defined relative to an external frame
of reference.  If the alignment of two distant reference frames deteriorates
with time, or if the state one possesses was prepared according to some
reference frame one does not have access to, super-selection rules are imposed
on the kind of states that can be prepared and the type of operations that can
be implemented~\cite{BW}.  For example, lacking a Cartesian frame of reference
associated with the rotation group, $\mathrm{SO}(3)$, implies that one can only
prepare rotationally invariant states and perform rotationally invariant
operations.  The lack of a reference frame can be alleviated (at least
partially) if a party is provided with a bounded-size quantum reference
frame~\cite{BRST}; i.e.~a non-invariant ($\mathrm{SO}(3)$-asymmetric) state,
such as the state of a particle with integer spin pointing in some known
direction~\cite{Peres2}. In the general case where the lack of a requisite
reference frame is associated with a finite or compact Lie group $G$, such
$G$-asymmetric states are considered \defn{resources}. They are measured by functions,
called ``frameness monotones''~\cite{GourSpekkens} or simply asymmetry
measures~\cite{Marvian,Gour}, that do not increase under the set of allowable
(i.e.~$G$-invariant) quantum operations.

There are several measures in the literature that quantify asymmetry for
particular groups~\cite{GourSpekkens}, or in particular dimensions~\cite{Borzu}.
However, there exists one measure, the $G$-asymmetry~\cite{Vaccaro}, that is
defined for all groups and all dimensions.  The \defn{$G$-asymmetry} of a state
$\rho\in\cB(\cH_d)$, a bounded, positive semi-definite operator acting on a
$d$-dimensional Hilbert space $\cH_d$, is defined as
\begin{equation}
A_G(\rho):=S\left(\cG\left[\rho\right]\right)-S(\rho),
\label{1}
\end{equation}
where $\cG(\cdot)\equiv\int\, \mathrm{d}g T(g)(\cdot)T(g)^{\dagger}$ is the \defn{$G$-twirling} operation, with the
integral performed over the Haar measure $\mathrm{d}g$, $T$ a representation of $G$ on the Hilbert space $\cH_d$, and $S$ the von Neumann entropy.  For the case of finite groups, the integral in the $G$-twirling operation is replaced by a sum, and the invariant Haar measure is given by $1/\lvert G\rvert$, where $\lvert G\rvert$ is the order of the group. The $G$-asymmetry was shown to be equal to the relative entropy of frameness~\cite{Gour}, the latter being
analogous to the relative entropy of entanglement~\cite{Plenio}. 

The relative entropy plays a crucial role in many quantum resource theories. Its importance comes from the fact that its asymptotic version provides the unique rate for reversible transformations~\cite{Horodecki}.  This property was demonstrated with the discovery that the regularized relative entropy of entanglement is the unique function that quantifies the rate of interconversion between states in a reversible theory of entanglement, 
where all types of non-entangling operations are allowed~\cite{Brandao08}. More 
recently, the importance of the relative entropy was demonstrated in the 
resource theory of thermodynamics~\cite{Horodecki11, Brandao11}. However, in the resource theory of reference frames the regularized relative entropy is 
zero on all states~\cite{Gour}. We will therefore need to rescale it in order to 
find its operational meaning.

In this work we consider the case where $G=\mathbb{Z}_M$, the cyclic group of $M$ elements, and the case where $G=U(1)$ associated with the important case of photon number super-selection. For these cases we provide an operational interpretation for the $G$-asymmetry that is information theoretic, different than the interpretation in~\cite{Vaccaro} of the $G$-asymmetry as extractable mechanical work. In particular, we find the strategy for aligning a pair of reference frames associated with $G$ that optimizes the amount of accessible information between the true relation, $g\in G$, of the two reference frames and the estimated relation, $g'\in G$, obtained by measurement. Our approach is different from previous work on reference frame alignment protocols where the success of the
protocol was quantified by maximizing an average cost, such as the fidelity or the maximum likelihood of a correct guess~\cite{Bagan1,Peres1,Bagan3,Peres2,Bagan2,Bagan4,Chiribella2,vonKorff}(for
a review of these protocols see ~\cite{review}). Choosing to deal only with cyclic abelian groups and $U(1)$ allows for easier exposition of our main result that, in the asymptotic limit, the amount of accessible information is equal to the $G$-asymmetry. As the accessible information is upper bounded by the Holevo quantity~\cite{Holevo},  our result also implies that for a reference frame alignment protocol associated with $G=\mathbb{Z}_M$ and $G=U(1)$, the Holevo bound is asymptotically achievable. 

This paper is organized as follows.  In Sec.~\ref{sec:2} we review the resource
theory of reference frames and the connection between the $G$-asymmetry, the
relative entropy of frameness, and the Holevo bound.  In Sec.~\ref{sec:3} we present our main results, and
determine the asymptotic rate of transmission of information in a phase alignment protocol (Sec.~\ref{sec:3a})~\cite{review}, and in a reference frame alignment protocol associated with $G=\mathbb{Z}_M$ (Sec.~\ref{sec:3b}) .  In Sec.~\ref{additivity} we discuss the additivity and superadditivity of the linearized, regularized $G$ asymmetry for both the $U(1)$ and $\mathbb{Z}_M$ cases.  We end the paper in Sec.~\ref{sec:4} with the summary and conclusions. 

\section{\label{sec:2} Relative entropy of frameness, the Holevo Bound, and the Alignment Rate}
Suppose two parties, Alice and Bob, wish to align their reference frames
associated with some group $G$.  Let $T:G\to\mathrm{GL}(\cH_d)$ be the unitary
representation of $G$ acting on a $d$-dimensional Hilbert space $\cH_d$.  Alice
prepares a system in a state $\ket{\psi}\in\cH_d$ and sends this system to Bob.
Assuming that Bob has complete ignorance as to which element $g\in G$ relates
his reference frame to Alice's, Bob's description of the state of the system 
sent to him by Alice is given by the ensemble $\{\mathrm{d}g,\,T(g)\ket{\psi}\}$, where $\mathrm{d}g$ is the Haar measure of the group $G$, or equivalently by the $G$-twirling of $\ket{\psi}$.
  
Alternatively, we can think of the above scenario in the following way.  We can assume that Alice
and Bob share a common frame of reference, but instead of sending the state $|\psi\rangle$, Alice sends Bob a state $T(g)|\psi\rangle$ picked from the ensemble
$\{\mathrm{d}g,\,T(g)\ket{\psi}\}$.  Bob's task is to determine $g$. In this view it is natural to quantify Bob's success in determining $g$ by the accessible information.

Let $X$ be the random variable consisting of the elements of $G$ with uniform probability distribution given by the Haar measure. Alice sends 
classical information to Bob by preparing a system in the state $\rho(X)=T(X)\rho T(X)^\dag$, where later we will assume that $\rho=\ketbra{\psi}$ is a pure state. 
Bob performs a positive operator valued measure (POVM) $\{E_y\}$, and obtains outcome $y$ with
probability $p_y$.  Let $Y$ denote the random variable associated with Bob's
measurement outcome. The \defn{accessible information}, $I^{(1)}(X:Y)$, for a single system is defined as the maximum amount of mutual information between random variables $X$ and $Y$, where the maximization is performed over all of Bob's possible POVMs.  An upper bound for the accessible information is the Holevo quantity~\cite{Holevo} which in the case of continuous groups is given by 
\begin{equation}
\chi^{(1)}=S\left(\cG[\rho]\right)-\int\,\mathrm{d}g\,
S\left(T(g)\rho T(g)^{\dagger}\right),
\label{2}
\end{equation}
where $S(\cdot)$ is the von Neumann entropy. If $G$ is a finite group, then the integral in Eq.~\eqref{2} is replaced by a sum, and $p_g=1/|G|$.  Since for any unitary transformation
$U, \,S(U\rho U^{\dagger})=S(\rho)$, the $G$-asymmetry given by Eq.~\eqref{1} is
equal to the Holevo quantity and thus is an upper bound for the accessible
information~\cite{Vaccaro06,Vaccaro12a,Gour}. For $N$ copies of a system prepared in the state $\rho$, the accessible information, $I^{(N)}(X:Y)$, is upper bounded by $\chi^{(N)}$, which is equal to the $G$-asymmetry, $A_G(\rho^{\otimes N})$, of the state $\rho^{\otimes N}$. 

The $G$-asymmetry is also related to the relative entropy of frameness defined
as 
\begin{equation}
\mathrm{min}_{\sigma\in\mathfrak{I}}\, S\left(\rho||\sigma\right),
\label{3}
\end{equation}
where $S\left(\rho||\sigma\right)=-S(\rho)-\mathrm{tr}\left(\rho\log\sigma\right)$,
and $\mathfrak{I}\equiv\{\sigma|\,\cG[\sigma]=\sigma\in\cB(\cH_d)\}$ is the set
of invariant states~\cite{Gour}.  Due to the invariance of the Haar measure, the $G$-twirling operation $\cG$, is idempotent, $\cG\circ\cG=\cG$, and unital,
$\cG[I]=I$.  For such a quantum operation it was shown in~\cite{Gour} that the
minimum relative entropy distance between an arbitrary state $\rho\in\cB(\cH)$
and a state $\sigma\in\mathfrak{I}$ satisfies
\begin{equation}
\mathrm{min}_{\sigma\in\mathfrak{I}}\,
S\left(\rho||\sigma\right)=S\left(\cG[\rho]\right)-S(\rho).
\label{4}
\end{equation}
Thus, the $G$-asymmetry is equal to the relative entropy of frameness.  

In entanglement theory, the regularized relative entropy of entanglement has an
operational interpretation. It quantifies the rate of interconversion between states 
in a reversible theory of entanglement~\cite{Horodecki}. One can therefore expect that the relative entropy of frameness would have an operational interpretation similar to that of the relative entropy of
entanglement.  However, it was shown in~\cite{Gour} that the regularized
relative entropy of frameness is zero on all states;
$$
\lim_{N\to\infty}\frac{A_{G}(\rho^{\otimes N})}{N}=0\;.
$$
This is because the
relative entropy of frameness is not an extensive quantity, i.e.~it does not scale linearly with the number of systems $N$, unlike the relative entropy of entanglement.  

For this reason, and following~\cite{Gour}, we introduce here the \emph{linearization function} $\mathcal{L}$. The linearization function $\mathcal{L}:\mathbb{R}\rightarrow\mathbb{R}$
is a monotonically increasing function that linearizes $A_G(\rho^{\otimes N})$; that is, $\mathcal{L}$ is chosen such that $\mathcal{L}(A_G(\rho^{\otimes N}))\propto N$ in the limit of large $N$. With such a choice for $\mathcal{L}$, we can rescale both the accessible information and the $G$-asymmetry so that they become extensive quantities.
We therefore define the \emph{regularized}, linearized $G$-asymmetry as follows:
\begin{equation}\label{regu}
A_{G}^{(reg)}(\psi)\equiv\lim_{N\rightarrow\infty}\frac{\mathcal{L}\left(A_G\left(\ketbra{\psi}^{\otimes N}\right)\right)}{N}\;.
\end{equation}
We will show in the next section that the above quantity has an operational interpretation for $G=\mathbb{Z}_M$ or $G=U(1)$.
It measures the rate, $R_{G}(\psi)$, at which many copies of a bounded-size quantum reference frame, $\ket{\psi}$, elevate
the restrictions induced by the lack of a shared reference frame. We now give its precise definition.

\begin{definition}\label{rate}
The alignment rate of a state $\ket{\psi}$ is defined by 
$$
R_G(\psi)\equiv\lim_{N\rightarrow\infty}\frac{\mathcal{L}(I^{(N)}(X:Y))}{N}
$$
\end{definition}
That is, $R_G(\psi)$, measures the rate at which Bob learns about the orientation of Alice's reference frame from many copies of
a bounded-size quantum reference frame, $\ket{\psi}$. By definition, $R_{G}(\psi)$ is a deterministic frameness monotone~\cite{GourSpekkens} (i.e.~does not increase under deterministic $G$-invariant operations). Furthermore, since $\mathcal{L}$ is a monotonic function, the Holevo bound implies that $R_{G}(\psi)$ is bounded above by $A_{G}^{(reg)}(\psi)$.
In the next section we show that in fact $R_G(\psi)=A_{G}^{(reg)}(\psi)$ for the case where $G=U(1)$ (Sec.~\ref{sec:3a}), associated with the case where two parties lack a common phase reference, and for the case where the reference frame is associated with the finite cyclic group of $M$ elements, $\mathbb{Z}_M$ (Sec.~\ref{sec:3b}).         
     
\section{\label{sec:3} Optimal rate for alignment of reference frames}   
In this section we determine the optimal alignment rate for the cases where Alice and Bob lack a shared frame of reference associated with $G=U(1)$ and  $G=\mathbb{Z}_M$.   Specifically, we show that for a phase reference ($G=U(1)$) $R_{U(1)}(\psi)$ is proportional to the \emph{number variance} of the bounded-size token of Alice's phase reference, whereas $R_{\mathbb{Z}_M}(\psi)$ is not proportional to variance of the bounded-size token of Alice's frame of reference even in the limit $M\to\infty$.

\subsection{\label{sec:3a} Phase reference}
We begin by considering two parties, Alice and Bob, who share an ideal quantum channel but lack a shared phase reference as discussed in Sec.~\ref{sec:2}.  The relevant group of transformations associated with a phase reference  is $U(1)$, the group of real numbers modulo $2\pi$.  Physically, Alice and Bob lack a shared  phase reference if their local oscillators have an unknown relative phase, i.e.~they are not phase locked.  The unitary representation, $T$, describing a phase shift $\theta\in U(1)$ is given by $T(\theta)=e^{\imath\theta\hat{N}}$, where $\hat{N}$ is the number operator.  If Alice and Bob have complete ignorance as to the relative phase between their respective local oscillators, then any state $\rho\in\cB(\cH)$  prepared by Alice is described as $\cG\left[\rho\right]$ by Bob.   

As the number operator is unbounded (from above), $T$ acts on an infinite dimensional Hilbert space, $\cH$.   Using Schur's lemmas~\cite{Sternberg}, the representation $T$ can be decomposed into the one-dimensional irreducible representations (irreps), $T^{(n)}(\theta)=e^{\imath \theta n}$ of $U(1)$,
\begin{equation}
T(\theta)=\bigoplus_{n=0}^{\infty} \alpha_n\, T^{(n)}(\theta),
\label{5}
\end{equation}
where the irrep label $n$ represents the total photon number, and $\alpha_n$ is the integer multiplicity of irrep $T^{(n)}$.    

Without a shared phase reference Alice faces a restriction on the type of states she can prepare with respect to Bob's phase reference. That is, the lack of a shared phase reference with Bob, imposes a photon number super-selection rule in which superpositions of photon number eigenstates cannot be prepared by Alice~\cite{review}.  Consequently, it is convenient to write the total Hilbert space, $\cH$, as
\begin{equation}   
\cH=\bigoplus_{n=0}^{\infty}\,\cH^{(n)}=\bigoplus_{n=0}^{\infty} \cM^{(n)}\otimes \cN^{(n)},
\label{6}
\end{equation}
where $\cM^{(n)}$ carries the irrep $T^{(n)}$, and $\cN^{(n)}$ carries the trivial representation of $U(1)$.   

In addition to preparation of states, a photon number super-selection rule also imposes restrictions on the types of operations  Alice can perform (relative to Bob's reference frame). More precisely, Alice is restricted to $U(1)$-invariant operations, which were shown in~\cite{GourSpekkens} to be of two types: shifts in the total photon number (by adding or removing photons), and changes in the relative amplitudes of different photon number states. In particular, the set of $U(1)$-invariant \emph{reversible} transformations consists of all unitary matrices that commute with the number operator  and shifts (in photon number). 

Thus, under a photon number super-selection rule all states $\ket{n,\alpha}\in\cM^{(n)}\otimes\cN^{(n)}$, for a given total photon number $n$, are equivalent up to U(1)-invariant unitary transformations. Hence, we can pick any pure state, say $\ket{n,\alpha=1}\in\cM^{(n)}\otimes\cN^{(n)}$, as our standard one. Hence,  any qudit $|\psi\rangle\in\cH$, can be brought by $U(1)$-invariant unitary transformations (and shifts), to a standard form 
\begin{equation}
\ket{\psi}=\sum_{n=0}^{d-1}\sqrt{p_n}\ket{n},
\label{7} 
\end{equation}
where $\sum_n p_n=1$.  

In the phase alignment protocol we consider here, Alice sends $N$ copies of a qudit prepared in the state $|\psi\rangle$ of Eq.~\eqref{7} to Bob.  The state $|\psi\rangle^{\otimes N}$ is a superposition of the \defn{tensor product number states} $\{\ket{x_1\ldots x_N},\; x_1,\ldots, x_N\in(0,\ldots,d-1)\}$.  Each such state can be written in terms of the total photon number, $n$, and its multiplicity $\alpha$. That is, $\ket{n,\alpha}\equiv\ket{x_1\ldots x_N}$ and $\alpha=1,...,l_n$, where $l_n$ denotes the number of orthonormal states with the same photon number $n$. 
In the basis $|n,\alpha\rangle$, the state $\ket{\psi}^{\otimes N}$ can be written as
\begin{equation}
\ket{\psi}^{\otimes N}=\sum_{n=0}^{N(d-1)}\sum_{\alpha=1}^{l_{n}}
\sqrt{c_{n,\alpha}}\,\,\ket{n,\alpha},
\label{8}
\end{equation}
where $c_{n,\alpha}$ are of the form $c_{n,\alpha}=\Pi_{j=0}^{d-1}p_j^{r_j}$, and $r_j$ are non-negative integers (corresponding to the number of times label $j$ appears in the string $x_1\ldots x_N$) satisfying $\sum_j r_j=N$ and $\sum_j jr_j=n$. Moreover, by $U(1)$-invariant unitary operations we can transform
\begin{equation}
\frac{\sum_{\alpha}\sqrt{c_{n,\alpha}}\ket{n,\alpha}}{\sqrt{\sum_{\alpha}c_{n,\alpha}}},
\label{9}
\end{equation}
to the standard state $\ket{n}\in\cM^{(n)}\otimes \cN^{(n)}$. This transformation brings $\ket{\psi}^{\otimes N}$
to the form
\begin{equation}
\ket{\psi}^{\otimes N}=\sum_{n=0}^{N(d-1)} \sqrt{c_{n}}\,\ket{n},
\label{10}
\end{equation}
where (see~\cite{Schuch,GourSpekkens})
\begin{equation}
c_n=\sum\binom{N}{r_0\ldots r_{d-1}}p_0^{r_0}\ldots p_{d-1}^{r_{d-1}}
\label{11}
\end{equation}
are the multinomial coefficients that arise from  the expansion $\left(\sum_n\sqrt{p_n}\ket{n}\right)^{\otimes N}$,
where terms giving rise to the same total photon number $n$ are grouped together.  We note that the sum in Eq.~\eqref{11} is taken over all non-negative integers $r_j$ for which $\sum_{j=0}^{d-1}r_j=N$ and $\sum_{j=0}^{d-1}jr_j=n$.  

Bob's description of the $N$ qudits sent to him by Alice is given by 
\begin{equation}
\ket{\psi(\theta)}^{\otimes N}=(T(\theta)\ket{\psi})^{\otimes N}=\sum_{n=0}^{N(d-1)}\sqrt{c_n}\,e^{\imath n\theta}\ket{n}.
\label{12}
\end{equation}
Due to the lack of a shared phase reference, Bob has complete ignorance about the value of $\theta$.   As the super-selection rule does not forbid us from performing any unitary on the multiplicity spaces, the state $\ket{\psi(\theta)}^{\otimes N}$ can be embedded in a $(N(d-1)+1)$-dimensional Hilbert space. This is the key reason why the $U(1)$-asymmetry is not an extensive quantity.  Bob's task is, therefore, to extract information about $\theta$ from a state in an $N(d-1)+1$-dimensional Hilbert space instead of a $d^N$-dimensional Hilbert space.

Suppose Bob's POVM is given by $\{E_{\theta'}\,\mathrm{d}\theta'\}$, where $E_{\theta'}\geq 0$ with $\int_0^{2\pi}E_{\theta'}\mathrm{d}\theta'=I=\sum_{n=0}^{N(d-1)}\ketbra{n}$. 
We would like to quantify how much Bob learns about $\theta$ from such a measurement.  Denote  by $\Theta$ the random variable associated with the relative phase, $\theta$, between Alice's and Bob's phase references.
That is, $\Theta=\theta$ with uniform probability distribution $p_\theta=1/2\pi$. Denote also by $\Theta'$ the random variable associated with Bob's measurement outcome $\theta'$.  Then, as discussed in the previous section,  the accessible information, $I^{(N)}(\Theta:\Theta')$, satisfies
\begin{equation}
I^{(N)}(\Theta:\Theta')\leq A_{U(1)}\left((\ketbra{\psi})^{\otimes N}\right).
\label{13}
\end{equation}
Using Eqs.~(\ref{1},~\ref{12}), the right-hand side of Eq.~\eqref{13} is
\begin{align}\nonumber
A_{U(1)}\left((\ketbra{\psi})^{\otimes N}\right)&=S\left(\cG\left[\ketbra{\psi}^{\otimes N}\right]\right)-S\left(\ketbra{\psi}^{\otimes N}\right)\\
&=S\left(\sum_{n=0}^{N(d-1)}\,c_n\ketbra{n}\right)= H\left(\{c_n\}\right),
\label{14}
\end{align}
where $H(\{c_n\})$ is the Shannon entropy of the probability distribution $\{c_n\}$.
Consequently, the accessible information per copy of the state $\ket{\psi}\in\cH_d$ obeys
\begin{equation}
\frac{I^{(N)}(\Theta:\Theta')}{N}\leq \frac{H\left(\{c_n\}\right)}{N}.
\label{15}
\end{equation}

We are interested in determining the accessible information per copy in the limit where the number of copies $N\rightarrow\infty$.  However, in this limit the right-hand side of Eq.~\eqref{15} tends to zero~\cite{Gour}.  Indeed, so long as the photon number spectrum is gapless, i.e.~$p_n\neq0$ for $0<n<d-1$~\cite{Schuch}, the probability distribution, $\{c_n\}$, can be approximated, in the limit $N\rightarrow\infty$,  with the normal distribution~\cite{Durrett}
\begin{equation}
c_n=\frac{1}{\sqrt{2\pi\sigma_N^2}}\exp\left(-\frac{(n-\mu_N)^2}{2\sigma_N^2}\right)+O\left(\frac{1}{N}\right),
\label{16}
\end{equation} 
where 
\begin{eqnarray}\nonumber
\sigma_N^2&=&N V(\psi)\equiv N\left(\sum_{n=0}^{d-1}n^2p_n-\left(\sum_{n=0}^{d-1}np_n\right)^2\right)\\ 
\mu_N&=&N\sum_{n=0}^{d-1}np_n,
\label{17}
\end{eqnarray}
with $V(\psi)$ the number variance of the state $\ket{\psi}$.  Using Eq.~\eqref{16}  the right-hand side of Eq.~\eqref{14} reads (see~\cite{Gour} for details)
\begin{equation}
\frac{H\left(\{c_n\}\right)}{N}=\frac{\frac{1}{2}\log(4\pi N V(\psi))+O\left(\frac{1}{\sqrt{N}}\right)}{N}.
\label{18}
\end{equation}
Due to the logarithmic dependence of $H(\{c_n\})$ on $N$ the $U(1)$-asymmetry is not an extensive quantity, and as a result the limit $N\rightarrow\infty$ of Eq.~\eqref{18} tends to zero.  

Following our discussion in the previous section, we introduce the linearization function, $\mathcal{L}(x)=2^{2x}$, so that the regularized $U(1)$-asymmetry is given by 
\begin{equation}
A_{U(1)}^{(reg)}(\psi)=\lim_{N\rightarrow\infty}\frac{\mathcal{L}\left(A_{U(1)}\left(\ketbra{\psi}^{\otimes N}\right)\right)}{N}=4\pi V(\psi).
\label{20}
\end{equation}
Furthermore, as $\mathcal{L}(x)=2^{2x}$ is a monotonically increasing function, it follows from Eq.~(\ref{15}) that Eq.~\eqref{20} is an upper bound for the rate of accessible information 
\begin{equation}
\frac{\mathcal{L}(I^{(N)}(\Theta:\Theta'))}{N}\leq\frac{\mathcal{L}(A_{U(1)}\left(\ketbra{\psi}^{\otimes N}\right)}{N}.
\label{21}
\end{equation}
In the following theorem we show that  in the limit $N\rightarrow\infty$, the inequality in Eq.~\eqref{21} is saturated, meaning that, for $G=U(1)$, the alignment rate as defined in Definition~\ref{rate} is equal to the regularized, linearized $G$-asymmetry.
\newline
\begin{theorem} \label{111} For $G=U(1)$
\[R_{G}(\psi)=A_{G}^{(reg)}(\psi)=4\pi V(\psi),
\]
where, $V(\psi)$ is the number variance of the state $\ket{\psi}$.
\end{theorem}
\begin{proof}
Recall that the accessible information is the maximum mutual information, $H(\Theta:\Theta')$, over all possible POVMs.  Let Bob's POVM elements be given by  $E_{\theta'}=\ketbra{e_{\theta'}}$, where
\begin{equation}
\ket{e(\theta')}=\frac{1}{\sqrt{2\pi}}\sum_{n=0}^{N(d-1)}\, e^{\imath n\theta'}\ket{n}.
\label{22}
\end{equation}
Note that $\int\,\mathrm{d}\theta'\, E_{\theta'}=I^{(N)}$, where $I^{(N)}=\sum_{n=0}^{N(d-1)}\ketbra{n}$.  We remark that the measurement of Eq.~\eqref{22} has been shown to optimize the average maximum likelihood~\cite{Chiribella1} and the average fidelity~\cite{review}.  We will show that in the limit $N\to\infty$ this measurement also maximizes $H^{(N)}(\Theta:\Theta')$ given by 
\begin{equation}
H^{(N)}(\Theta:\Theta')=\int_{0}^{2\pi}\mathrm{d}\theta\int_{0}^{2\pi}\mathrm{d}\theta' p(\theta,\theta')\,\log\left(\frac{p(\theta,\theta')}{p(\theta')\,p(\theta)}\right),
\label{23}
\end{equation}
where the joint probability distribution, $p(\theta,\theta')$ can be calculated using Bayes' rule, $p(\theta,\theta')=p(\theta'|\theta)\,p(\theta)$.  In our case
\begin{eqnarray}\nonumber
p(\theta)&=&\frac{1}{2\pi}\\
\label{24}
p(\theta'|\theta)&=&\left|\bracket{e(\theta')}{\psi(\theta)}^{\otimes N}\right|^2
\end{eqnarray}
Substituting Eqs.~(\ref{12},~\ref{22}) into Eq.~\eqref{24} gives
\begin{align}\nonumber
p(\theta'|\theta)&=\frac{1}{2\pi}\sum_{n,m=0}^N\sqrt{c_nc_m}\,e^{
\left(\imath(m-n)(\theta-\theta')\right)}\\
&=\frac{1}{2\pi}\left\lvert\sum_{m=0}^N\sqrt{c_m}e^{\imath m(\theta-\theta')}\right\rvert^2
\label{25}
\end{align}
From the equation above we see that $p(\theta'|\theta)=p(\theta|\theta')$. Therefore, the probability that $\Theta'=\theta'$ is given by  
$$
p(\theta')=\int_0^{2\pi}\, p(\theta'|\theta)p(\theta)\mathrm{d}\theta=\frac{1}{2\pi}\int_0^{2\pi}\, p(\theta|\theta')\mathrm{d}\theta=\frac{1}{2\pi}\;.
$$ 
Hence, Eq.~\eqref{23} reduces  to 
\begin{equation}
H^{(N)}(\Theta:\Theta')=\frac{1}{2\pi}\int_0^{2\pi}\mathrm{d}\theta\int_0^{2\pi}\mathrm{d}\theta' p(\theta'|\theta)\log(2\pi p(\theta'|\theta)).
\label{26}
\end{equation}

The expression for the conditional probabilities $p(\theta'|\theta)$ can be greatly simplified as we now show.  Note that the sum in Eq.~\eqref{25} runs over positive integers. In the limit of large $N$ the sum can be approximated by an integral over a continuous variable $m$. Furthermore, as $\mu_N$ is large and positive, the probability distribution corresponding to small photon numbers lies at the tail end of the Gaussian. Using the properties of the error function the lower limit of integration can be extended to negative photon numbers, accumulating a negligible ($O(N^{-1})$) total probability.  Making a change of variable, $\phi=\theta-\theta'$, and using Eq.~\eqref{16}, Eq.~\eqref{25} becomes 
\begin{align}\nonumber
p_{\phi}\equiv p(\theta|\theta')&=\frac{1}{\sqrt{2\pi\sigma_N^2}}\left\lvert\int_{-\infty}^{\infty}\frac{1}{\sqrt{2\pi}}\,e^{\left(-\frac{(m-\mu_N)^2}{4\sigma_N^2}\right)}e^{\imath m\phi}\mathrm{d}m\right\rvert^2\\ \nonumber
&+O\left(\frac{1}{\sqrt{N}}\right)\\
&=\sqrt{\frac{2\sigma_N^2}{\pi}}e^{-2\phi^2\sigma_N^2}+O\left(\frac{1}{\sqrt{N}}\right).
\label{27}
\end{align}
Using Eq.~\eqref{27} and noting that 
\begin{equation}
\int_0^{2\pi}\mathrm{d}\theta\int_o^{2\pi}\mathrm{d}\theta'=\int_{-2\pi}^{2\pi}\mathrm{d}\phi\int_{\phi}^{2\pi}\mathrm{d}\theta,
\label{29}
\end{equation}   
Eq.~\eqref{26} reduces to 
\begin{align}\nonumber
H^{(N)}(\Theta:\Theta')&=\sqrt{\frac{2\sigma_N^2}{\pi}}\int_{-2\pi}^{2\pi}\mathrm{d}\phi e^{-2\phi^2\sigma_N^2}\\
&\times\log\left(\sqrt{8\pi\sigma_N^2}\,e^{-2\phi^2\sigma_N^2}\right)+O\left(\frac{1}{\sqrt{N}}\right).
\label{30}
\end{align}
We note that the mutual information does not depend on the mean photon number  as expected, since the latter can be shifted using $U(1)$-invariant operations, and therefore cannot carry any phase information. 

Using the approximations 
\begin{align}\nonumber
& \int_{-2\pi}^{2\pi}\mathrm{d}\phi\,e^{-2\phi^2\sigma_N^2}=\int_{-\infty}^{\infty}\mathrm{d}x\,e^{-2\sigma_N^2 x^2}+O\left(\frac{1}{N}\right)\\\nonumber
& \int_{-2\pi}^{2\pi}\mathrm{d}\phi\, \phi^2\,e^{-2\phi^2\sigma_N^2}=\int_{-\infty}^{\infty}\mathrm{d}x\,x^2\,e^{-2\sigma_N^2 x^2}+O\left(\frac{1}{N}\right),\\
\label{31}
\end{align}
where the integrals on the right-hand side of Eq.~\eqref{31} are equal to $\sqrt{\pi/2\sigma_N^2}$, and $1/2\sqrt{\pi/8\sigma_N^6}$, respectively, one obtains, after simple algebra, 
\begin{equation}
H^{(N)}(\Theta:\Theta')=\frac{1}{2}\log(4\pi\sigma_N^2)+O\left(\frac{1}{\sqrt{N}}\right).
\label{32}
\end{equation}
Finally, linearizing the accessible information and taking the limit gives
\begin{equation}
\lim_{N\rightarrow\infty}\frac{\mathcal{L}(I^{(N)}(\Theta:\Theta')}{N}=\lim_{N\rightarrow\infty}\frac{2^{\log(4\pi N V(\psi))}}{N}=4\pi V(\psi).
\label{33}
\end{equation}
This completes the proof.
\end{proof}

\subsection{\label{sec:3b} Reference frame associated with $\mathbb{Z}_M$}
We now consider the case where Alice and Bob share an ideal quantum channel but lack a shared reference frame associated with the finite cyclic group of $M$ elements, $\mathbb{Z}_M$.  For example, the case $G=\mathbb{Z}_2$ corresponds to the situation where Alice and Bob lack a reference frame for chirality~\cite{GourSpekkens}.   
Unlike the $U(1)$-case, we find that in this case the optimal rate for the alignment of reference frames is not proportional to the variance even in the limit $M\to\infty$. This is not inconsistent with Theorem~\ref{111} of the previous subsection since our main assumption here is that $N\gg M$. Therefore, the results obtained in this subsection are completely independent on the previous subsection.

The unitary representation $T(g)$ (with $g\in\mathbb{Z}_M$) acting on the Hilbert space $\cH$ can be decomposed into one-dimensional irreps $T^{(k)}$ as 
\begin{equation}
T(g)=\bigoplus_{k=0}^{M-1}\alpha_k T^{(k)}(g).
\label{34}
\end{equation}
where $k$ labels the irreps of $\mathbb{Z}_M$, and $\alpha_k$ is the multiplicity of irrep $T^{(k)}$.  Just as in the $U(1)$ case above the lack of a shared reference frame associated with $\mathbb{Z}_M$ imposes restrictions on the type of states Alice can prepare with respect to Bob's reference frame.  In order to describe these restrictions, it is convenient to write the total Hilbert space, $\cH$, as
\begin{equation}
\cH=\bigoplus_{k=0}^{M-1}\cH^{(k)}=\bigoplus_{k=0}^{M-1}\cM^{(k)}\otimes \cN^{(k)},
\label{35}
\end{equation}
where $\cM^{(k)}$ is the carrier space of $T^{(k)}$, and $\cN^{(k)}$ carries the trivial representation of $\mathbb{Z}_M$.
Thus, the lack of a shared reference frame, associated with $\mathbb{Z}_M$, imposes a  super-selection rule in which superpositions of states from different sectors $\cH^{(k)}$ cannot be prepared. Note that unlike the $U(1)$ case above,  there are a finite number of sectors, $\cH^{(k)}$, equal to the order of the group.
      
In addition to preparation of states, a $\mathbb{Z}_M$ super-selection rule also imposes restrictions on the types of operations Alice can perform  (relative to Bob's reference frame). More precisely, Alice is restricted to $\mathbb{Z}_M$-invariant operations.  In the case were Alice and Bob lack a chiral frame of reference (associated with $\mathbb{Z}_2$) it was shown in~\cite{GourSpekkens} that $\mathbb{Z}_2$-invariant operations are of two types:  shifts in the irrep label $k$ (which in the case of chiral frames corresponds to the bit flip operation, $X$),  and changes in the relative amplitudes of different eigenstates of irrep label $k$.  Similarly, in the case of a $\mathbb{Z}_M$ super-selection rule, the $\mathbb{Z}_M$-invariant operations consist of shifts (mod $M$) in the irrep label $k$, and changes in the relative amplitudes of different eigenstates of irrep label $k$. 

Thus, any qudit $\ket{\psi}\in\cH$, can be brought by $\mathbb{Z}_M$-invariant unitary transformations (and shifts), to a standard form
\begin{equation}
\ket{\psi} =\sum_{k=0}^{M-1}\sqrt{p_k}\ket{k},
\label{36}
\end{equation} 
where $\sum_kp_k=1$, and $\ket{k}\equiv\ket{k,\alpha=1}\in\cM^{(k)}\otimes\cN^{(k)}$ is a state in $\cM^{(k)}\otimes\cN^{(k)}$ chosen to be the standard one.  This is because under the $\mathbb{Z}_M$ super-selection rule all states $\ket{k,\alpha}\in\cM^{(k)}\otimes\cN^{(k)}$, for a given irrep label $k$, are equivalent up to $\mathbb{Z}_M$-invariant unitary transformations.  Hence, we can pick any pure state, say $\ket{k,\alpha=1}\in\cM^{(k)}\otimes\cN^{(k)}$, as our standard one. 

In the reference frame alignment protocol we consider here, where the reference frame is associated with $\mathbb{Z}_M$,   Alice sends $N$ copies of a qudit  prepared in the state $\ket{\psi}$ of Eq.~\eqref{36} to Bob.  The state $\ket{\psi}^{\otimes N}$ is a superposition of the tensor product basis $\{\ket{x_1\ldots x_N},\; x_1,\ldots, x_N\in(0,\ldots,d-1)\}$.  Each such state can be written in terms of the irrep label $k$, and its multiplicity $\alpha$.  That is, $\ket{k,\alpha}\equiv\ket{x_1\ldots x_N}$  and $\alpha=1,\ldots,l_k$, where $l_k$ denotes the number of orthonormal states with the same irrep label $k$.  
In the basis $\ket{k,\alpha}$, the state $\ket{\psi}^{\otimes N}$ can be written as
\begin{equation}
\ket{\psi}^{\otimes N}=\sum_{k=0}^{M-1}\sum_{\alpha=1}^{l_k}\sqrt{c_{k,\alpha}}\,\,\ket{k,\alpha},
\label{38}
\end{equation}
where $c_{k,\alpha}$ are of the form $c_{k,\alpha}=\Pi_{j=0}^{M-1}p_j^{r_j}$, and $r_j$ are positive integers (corresponding to the number of times $x_j\in(0,\ldots,d-1)$ appears in $\ket{k,\alpha}$) satisfying $\sum_j r_j=N$ and $\left(\sum_j jr_j\right)_{\mathrm{mod}\, M}=k$.  Moreover, by $\mathbb{Z}_M$-invariant unitary operations we can transform 
\begin{equation}
\frac{\sum_{\alpha}\sqrt{c_{k,\alpha}}\ket{k,\alpha}}{\sqrt{\sum_\alpha c_{k,\alpha}}},
\label{39}
\end{equation} 
to the standard state $\ket{k}\in\cM^{(k)}\otimes\cN^{(k)}$. This transformation brings $\ket{\psi}^{\otimes N}$ to the form
\begin{equation}
\ket{\psi}^{\otimes N}=\sum_{k=0}^{M-1}\sqrt{c_k}\ket{k},
\label{40}
\end{equation}
where
\begin{equation}
c_k=\sum\binom{N}{r_0\ldots r_{M-1}}p_0^{r_0}\ldots p_{M-1}^{r_{M-1}},
\label{41}
\end{equation}
are the multinomial coefficients that arise from the expansion $\left(\sum_k\sqrt{p_k}\ket{k}\right)^{\otimes N}$,
where terms giving rise to the same irrep label $k$ are grouped together.  We note that the sum in Eq.~\eqref{41} is taken over integers $r_j$ for which $\sum_j r_j=N$ and $\left(\sum_j jr_j\right)_{\mathrm{mod}\,M}=k$.    

Note that Eq.~\eqref{41} is similar to Eq.~\eqref{11} in the $U(1)$ case above with the important difference that $\sum_j jr_j$ is modulo $M$. As we are considering finite cyclic groups $M<<N$, and in the limit $N\rightarrow\infty$ the probability distribution $\{c_k\}$ can no longer be approximated with the normal distribution.

The coefficients $c_k$ can also be written as 
\begin{equation}
c_k=\sum_{m_1=0}^{M-1}\ldots\sum_{m_N=0}^{M-1}\delta_{\bar{m},k}p_{m_1}\ldots p_{m_N},
\label{42}
\end{equation}
where $\bar{m}=\sum_i m_i$.
In order to simplify calculations involving the discreet probability coefficients $\{c_k\}$, we use the discreet Fourier transform to re-write Eq.~\eqref{42} as  
\begin{equation}
c_k=\frac{1}{M}\sum_{n=0}^{M-1}e^{-\frac{\imath2\pi kn}{M}}z_n,
\label{43}
\end{equation}
where 
\begin{equation}
z_n\equiv\left(\sum_{m=0}^{M-1}e^{\frac{\imath2\pi nm}{M}}\,p_m\right)^N\equiv(r_ne^{\imath\theta_n})^N,
\label{44}
\end{equation}
with $0<r_n\leq1$ and the phase $\theta_n\in[0,2\pi)$. Since $z_0=1$, Eq.~\eqref{43} can be written as
\begin{equation}
c_k=\frac{1}{M}\left(1+\sum_{n=1}^{M-1}e^{-\frac{\imath2\pi kn}{M}}z_n\right)\equiv\frac{1}{M}(1+\Delta_k),
\label{45}
\end{equation}
where $\Delta_k$ must be real since $c_k$ are real.  Moreover, using the triangle inequality $\lvert\Delta_k\rvert\leq\sum_{n=1}^{M-1}\lvert z_n\rvert$.  As $1\leq n\leq M-1$, and $\sum_{m=0}^{M-1} p_m=1$, where $p_m<1,\,\forall m\in(0,\ldots,M-1)$, there exists $0< s_n<1$ such that $\lvert r_ne^{\imath\theta_n}\rvert<s_n$.
Therefore, $\lvert z_n\rvert<s_n^N$.  Denoting $s_{\mathrm{max}}\equiv\mathrm{max}\{s_n\}$ it follows that $\left\lvert\Delta_k\right\rvert\leq(M-1)s_{\mathrm{max}}^N$, 
and in the limit $N\rightarrow\infty$, $\lvert\Delta_k\rvert$ goes exponentially to zero for all $k$, 
which also implies that as $N\rightarrow\infty$, $c_k\rightarrow 1/M$.  

Indeed the set of states   
\begin{equation}
\left\lbrace T(g)\ket{+}=\frac{1}{M}\sum_{k=0}^{M-1}e^{\frac{\imath2\pi kg}{M}}\ket{k}\;\Big|\; g=0,\ldots, M-1\right\rbrace,
\label{47}
\end{equation}
where $\ket{+}\equiv T(g=0)\ket{+}$, are optimal resources if Alice and Bob lack a shared frame of reference for $\mathbb{Z}_M$.  Bob can perfectly distinguish the states in Eq.~\eqref{47} and learn Alice's reference frame.   For example, if Alice and Bob lack a chiral frame, associated with $\mathbb{Z}_2$, then the states $\ket{\pm}=1/\sqrt{2}(\ket{0}\pm\ket{1})$, encode all the information about Alice's reference frame.  If Bob detects $\ket{+}$ then he knows that his and Alice's chiral frames are aligned, else they are anti-aligned.  

Bob's description of the $N$ qudits sent to him by Alice is given by 
\begin{equation}
\ket{\psi(g)}=(T(g)\ket{\psi})^{\otimes N}=\sum_{k=0}^{M-1}\sqrt{c_k}\,e^{\frac{\imath2\pi kg}{M}}\ket{k}.
\label{48}
\end{equation}  
Due to the lack of a shared reference frame, Bob has complete ignorance about the element $g\in \mathbb{Z}_M$. As the super-selection rule does not forbid us from performing any unitary on the multiplicity spaces, the state $\ket{\psi(g)}^{\otimes N}$ can be embedded in a $M$-dimensional Hilbert space.  Bob's task is, therefore, to extract information about $g\in\mathbb{Z}_M$ from a state in an $M$-dimensional Hilbert space instead of a $d^N$-dimensional Hilbert space. 

Suppose Bob's POVM is given by $\{E_{y},\,y\in\mathbb{Z}_M\}$, where $E_{y}\geq 0$, with $
\sum_{y\in \mathbb{Z}_M} E_{y}=I=\sum_{k=0}^{M-1}\ketbra{k}$.  We would like to quantify how much 
Bob learns about $g\in\mathbb{Z}_M$ from such a measurement.  Denote by $X$ the random variable 
associated with the relative group element, $x\in\mathbb{Z}_M$, between Alice's and Bob's reference frames.  That is $X=x$  with uniform probability distribution $p_x=1/M$. Denote also by $Y$ the random variable associated with Bob's measurement outcome, $y\in\mathbb{Z}_M$.  Using the same reasoning as in Sec.~\ref{sec:3a} the accessible information per copy obeys
\begin{equation}
\frac{I_{\mathbb{Z}_M}^{(N)}(X:Y)}{N}\leq \frac{H\left(\{c_k\}\right)}{N}.
\label{49}
\end{equation}
where $H(\{c_k\})$ is the Shannon entropy of the probability distribution $\{c_k\}$.  Using Eq.~\eqref{45}, the latter reads
\begin{equation}
H\left(\{c_k\}\right)=-\frac{1}{M}\sum_{k=0}^{M-1}(1+\Delta_k)\log\left(\frac{1}{M}(1+\Delta_k)\right).\\
\label{50}
\end{equation} 
As $\Delta_k$ are small, we can use the Taylor expansion for the logarithm. Thus, noting that $\sum_{k=0}^{M-1}\Delta_k=0$, Eq.~\eqref{50} can be written as 
\begin{equation}
H\left(\{c_k\}\right)=\log M-\frac{1}{M\ln2}\sum_{k=0}^{M-1}\sum_{n=2}^\infty (-1)^n\frac{\Delta_k^n}{n(n-1)}.
\label{51}
\end{equation}
Note that $H\left(\{c_k\}\right)$ is equal to $\log M$ with a correction that, for large $N$, goes exponentially to zero
(recall that the $\Delta_k$'s go exponentially to zero). We now find out the dominant part of this correction.

First, note that 
\begin{align} \nonumber
\sum_k \Delta_k^2&=\sum_{k=0}^{M-1}\sum_{n,m=1}^{M-1}e^{\frac{\imath2\pi k(n+m)}{M}}z_nz_m=M\sum_{n=1}^{M-1}|z_n|^2\\ \nonumber
\sum_k \Delta_k^3&=\sum_{k=0}^{M-1}\sum_{n,m,l=1}^{M-1}e^{\frac{\imath2\pi k(n+m+l)}{M}}z_nz_mz_l\\
&=M\sum_{n,m=1}^{M-1}z_nz_mz_{M-(m+n)},
\label{52}
\end{align}
and similar expressions can be found for $\sum_k \Delta_k^n$ for $n>3$.
Writing the complex numbers $z_n$ as in Eq.~\eqref{44}, Eq.~\eqref{52} becomes
\begin{align}\nonumber
\sum_k \Delta_k^2&=M\sum_{n=1}^{M-1} r_n^{2N}\\
\sum_k \Delta_k^3&=M\nonumber\\
\times\sum_{n,m=1}^{M-1}&(r_n r_m r_{M-(m+n)})^N \cos\left( N(\theta_n+\theta_m+\theta_{M-(m+n)})\right).
\label{53}
\end{align}
As the sums in Eq.~\eqref{53} are over terms that are very small, we focus here only on the dominant terms with the maximum value of $r_n$.  We therefore define 
$$
S=\{l\;\big|\; r_l=r_{\max}\}\;\;;\;\;
r_{\max}\equiv\max_{n=1,...,M-1}\left|\sum_{m=0}^{M-1}e^{\frac{\imath2\pi nm}{M}}\,p_m\right|
$$ 
to be the set  of all integers, $l$, for which the magnitude of $z_l$ (see Eq.~\eqref{44}) is maximum. 

While the dominant terms in the first sum of Eq.~\eqref{53} are proportional to $r_{\max}^{2N}$, the second sum is  exponentially smaller than $r_{\max}^{2N}$.
Similarly, for any $n>2$ the sum $\sum_{k}\Delta_k^n$ is exponentially smaller than $r_{\max}^{2N}$.
Therefore, Eq.~\eqref{51} can be written as 
\begin{equation}
H\left(\{c_k\}\right)=\log M-r_{\mathrm{max}}^{2N}\left(\frac{\lvert S\rvert}{2\ln2}+O\left(\left(\frac{r}{r_{\mathrm{max}}}\right)^N\right)\right),
\label{57}
\end{equation}
where $r$ is some positive number smaller than $r_{\max}$, and $|S|$ denotes the size of $S$.  Note that the maximum $H\left(\{c_k\}\right)$ can be is $\log M$, and this maximum is achieved if and only if $\ket{\psi}$ in Eq.~\eqref{36} is one of the optimal resource states in Eq.~(\ref{47}). It follows that the regularized $\mathbb{Z}_M$-asymmetry goes to zero in the limit $N\to\infty$. 

Just as in the $U(1)$ case we need to modify the $\mathbb{Z}_M$-asymmetry so that it scales linearly in $N$.  This is achieved by defining the linearization function, $\mathcal{L}:\mathcal{R}\rightarrow\mathcal{R}$ to be~\footnote{Following the linearization introduced in~\cite{GourSpekkens} for the special case of $G=\mathbb{Z}_2$.} $$\mathcal{L}(x)=-\log(\log M-x).$$  As this is a monotonically increasing function it follows that  
{\begin{equation}
\frac{\mathcal{L}(I_{\mathbb{Z}_M}^{(N)}(X:Y))}{N}\leq\frac{\mathcal{L}\left(A_{\mathbb{Z}_M}\left((\ketbra{\psi})^{\otimes N}\right)\right)}{N}.
\label{58}
\end{equation}}
In the following theorem we show that in the limit $N\rightarrow\infty$, the inequality in Eq.~\eqref{58} is 
saturated. 
\newline
\begin{theorem} \label{thm:2}
Let $\ket{\psi} =\sum_{k=0}^{M-1}\sqrt{p_k}\ket{k}$ as in Eq.~\eqref{36}.  Then, for $G=\mathbb{Z}_M$
\[R_G(\psi)=A_{G}^{(reg)}(\psi)=-2\log r_{\mathrm{max}},
\]
where $$r_\mathrm{max}=\max_{n\in\{1,2,...,M-1\}}\left|\sum_{m=0}^{M-1}e^{\frac{\imath2\pi nm}{M}}\,p_m\right|$$.
\end{theorem}
\begin{proof}
Let Bob's POVM elements be given by $E_y=\ketbra{e_y}$, where
\begin{equation}
\ket{e(y)}=\frac{1}{\sqrt{M}}\sum_{k=0}^{M-1} e^{\frac{\imath2\pi ky}{M}}\ket{k}.
\label{59}
\end{equation}
Note that $\sum_{y\in\mathbb{Z}_M} E_y=I=\sum_{k=0}^{M-1}\ketbra{k}$. We will show that the measurement in Eq.~\eqref{59} maximizes $H^{(N)}(X:Y)$ given by 
\begin{equation}
H^{(N)}(X:Y)=\sum_{x,y}p(x,y)\log\left(\frac{p(x,y)}{p(x) p(y)}\right),
\label{60}
\end{equation}
where the joint probability distribution $p(x,y)$ can be calculated using Bayes' rule $p(x,y)=p(y|x)p(x)$, and 
\begin{align}\nonumber
p(x)&=\frac{1}{M}\\
p(y|x)&=\left|\bracket{e(y)}{\psi(x)}\right|^2 
\label{61}
\end{align}
Substituting  Eqs.~(\ref{48},\ref{59}) into Eq.~\eqref{61} gives
\begin{align} \nonumber
p(y|x)&=\frac{1}{M}\sum_{k,l=0}^{M-1}\sqrt{c_k\,c_l}e^{\frac{\imath2\pi(k-l)(x-y)}{M}}\\
&=\frac{1}{M}\left\lvert\sum_{k=0}^{M-1}\sqrt{c_k}e^{\frac{\imath2\pi k(x-y)}{M}}\right\rvert^2
\label{62}
\end{align}
From the equation above we see that $p(y|x)=p(x|y)$.  Therefore, the probability that $Y=y$ is given by 
\begin{equation}
p(y)=\sum_{x=0}^{M-1} p(y|x)p(x)=\frac{1}{M}\sum_{x=0}^{M-1}\, p(y|x)=\frac{1}{M}.
\end{equation}
Hence,  Eq.~\eqref{60} reduces to
\begin{equation}
H^{(N)}(X:Y)=\log M+\frac{1}{M}\sum_{x,y=0}^{M-1}p(y|x)\log(p(y|x)),
\label{63}
\end{equation}
and using Eq.~\eqref{45} the conditional probabilities, $p(y|x)$, maybe written as 
\begin{equation}
p(y|x)=\frac{1}{M^2}\sum_{k,l=0}^{M-1}e^{\frac{\imath2\pi(k-l)(x-y)}{M}}\sqrt{(1+\epsilon_{kl})},
\label{64}
\end{equation} 
where $\epsilon_{kl}=\Delta_k+\Delta_l+2\Delta_k\Delta_l$. As $\epsilon_{kl}$ is small and second order in $\Delta$, expanding the square root in Eq.~\eqref{64} to second order in $\epsilon_{kl}$ gives
\begin{align}\nonumber
p(y|x)&=\frac{1}{M^2}\sum_{k,l=0}^{M-1}e^{\frac{\imath2\pi(k-l)(x-y)}{M}}\left(1+\frac{1}{2}(\Delta_k+\Delta_l)\right.\\
&\left.-\frac{1}{8}(\Delta_k^2+\Delta_l^2)+\frac{1}{4}\Delta_k\Delta_l\right)+ O(\epsilon_{kl}^3).
\label{65}
\end{align} 
As $p(y|x)=p(x|y)$, and $\sum_{k=0}^{M-1}\Delta_k=0$,
\begin{equation}
\frac{1}{4M^2}\sum_{k,l=0}^{M-1}e^{\frac{\imath2\pi(x-y)(k-l)}{M}}\Delta_k\Delta_l=\left\{\begin{array}{l l}
                      0 & \, \mbox{if $x=y$}\\
                      \frac{1}{4}|z_{x-y}|^2 & \, \mbox{if x\textgreater y}\\ \end{array} \right.,
\label{66}
\end{equation} 
and Eq.~\eqref{65} can be written, after some algebra, as
\begin{align}\nonumber
p(x|x)&=1-\frac{1}{4M}\sum_{k=0}^{M-1}\Delta_k^2+O(\Delta_k^3)\\
p(y\neq x|x)&=\frac{1}{4}|z_{x-y}|^2.
\label{67}
\end{align}
Using the same arguments as in Eqs.~(\ref{51}-\ref{57}) above, Eq.~\eqref{67} can be written as 
\begin{align}\nonumber
p(x|x)&=1-\frac{r_{\mathrm{max}}^{2N}}{4}\left(\lvert S\rvert+O\left(\left(\frac{r}{r_{\mathrm{max}}}\right)^N\right)\right)\\
p(y\neq x|x)&=\frac{1}{4}|z_{x-y}|^2.
\label{68}
\end{align}

We now break the mutual information, Eq.~\eqref{63}, into two terms
\begin{align}\nonumber
H^{(N)}(X:Y)&=\log M+\frac{1}{M}\sum_{x=0}^{M-1}p(x|x)\log(p(x|x))\\
&+\frac{2}{M}\sum_{x>y}p(y|x)\log(p(y|x))
\label{69}
\end{align}
Using Eq.~\eqref{68}, the approximation $(1-x)\log(1-x)=\frac{1}{\ln2}(-x+O(x^2))$, and noting that the terms in the first sum of Eq.~\eqref{69} are independent of $x$, we obtain
\begin{align}\nonumber
&\frac{1}{M}\sum_{x=0}^{M-1}p(x|x)\log p(x|x)\\
&=-\frac{r_{\mathrm{max}}^{2N}}{4\ln2}\left(\lvert S\rvert
+O\left(\left(\frac{r}{r_{\mathrm{max}}}\right)^N\right)\right).
\label{70}
\end{align}
The second sum in Eq.~\eqref{69} reads
\begin{equation}
\frac{2}{M}\sum_{x>y}p(y|x)\log(p(y|x))=\frac{1}{2M}\sum_{x>y}|z_{x-y}|^2\log\left(\frac{|z_{x-y}|^2}{4}\right).
\label{71}
\end{equation}
Denoting $n=x-y$ and noting that $\sum_{x>y}=\sum_{n=1}^{M-1}\sum_{y=0}^{M-1-n}$, Eq.~\eqref{71} becomes
\begin{align}\nonumber
&\frac{2}{M}\sum_{x>y}p(y|x)\log(p(y|x))=\\
&\frac{1}{2M}
\sum_{n=1}^{M-1}(M-n)|z_n|^2\left(\log|z_n|^2-2\right).
\label{72}
\end{align}
Plugging Eqs.~(\ref{70},~\ref{72}) into Eq.~\eqref{69} gives 
\begin{align}\nonumber
H^{(N)}(X:Y)&=\log M-\frac{r_{\mathrm{max}}^{2N}}{4\ln2}\left(\lvert S\rvert+O\left(\left(\frac{r}{r_{\mathrm{max}}}\right)^N\right)\right) \nonumber\\
&+\frac{1}{2M}
\sum_{n=1}^{M-1}(M-n)|z_n|^2\left(\log|z_n|^2-2\right).
\label{73}
\end{align} 
As only the largest $|z_n|$'s will contribute significantly to the mutual information, Eq.~\eqref{73} reduces to
\begin{align}\nonumber
H^{(N)}(X:Y)&=\log M-r_{\mathrm{max}}^{2N}\left(\frac{|S|}{4\ln2}+\frac{1}{M}\sum_{s\in S}(M-s)\right.\\
&\left.+\frac{N\log(r_{\mathrm{max}})}{M}\sum_{s\in S} (M-s)+O\left(\left(\frac{r}{r_{\mathrm{max}}}\right)^N\right)\right),
\label{74}
\end{align} 
Denoting by $D\equiv\frac{\sum_{s\in S}(M-s)}{M}$, Eq.~\eqref{74} becomes
\begin{align}\nonumber
&H^{(N)}(X:Y)=\log M\\
&-r_{\mathrm{max}}^{2N}\left(\frac{|S|}{4\ln2}+D(1-N\log r_{\mathrm{max}})
+O\left(\left(\frac{r}{r_{\mathrm{max}}}\right)^N\right)\right)
\label{75}
\end{align}

Finally, linearizing both the mutual information and the $\mathbb{Z}_M$-asymmetry yields
\begin{align}\nonumber
\mathcal{L}(A_{\mathbb{Z}_M}\left((\ketbra{\psi})^{\otimes N}\right)&=-\log\left(\frac{|S|}{2\ln2}+O\left(\left(\frac{r}{r_{\mathrm{max}}}\right)^N\right)\right)\\  \nonumber
&-2N\log r_{\mathrm{max}}\\ \nonumber
\mathcal{L}(H^{(N)}(X:Y))&=-\log\left(\frac{|S|}{4\ln2}+D(1-N\log r_{\mathrm{max}})\right.\\ 
&\left.+O\left(\left(\frac{r}{r_{\mathrm{max}}}\right)^N\right)\right)-2N\log r_{\mathrm{max}}.
\label{76}
\end{align}
Dividing both quantities in Eq.~\eqref{76} by $N$ and taking the limit $N\rightarrow\infty$, one notes that the first term in both quantities tends to zero.  Thus
\begin{align}\nonumber
\lim_{N\rightarrow\infty}\frac{\mathcal{L}(A_{\mathbb{Z}_M}\left((\ketbra{\psi})^{\otimes N}\right)}{N}&=\lim_{N\rightarrow\infty}\frac{\mathcal{L}(H^{(N)}(X:Y))}{N}\\
&=-2\log r_{\mathrm{max}}
\label{77}
\end{align}
This completes the proof.
\end{proof}

Thus the alignment rate is equal to the regularized, linearized $\mathbb{Z}_M$-asymmetry.  In the next section we show, somewhat surprisingly, that the additivity of $R_{G}(\psi)$ does not hold for all finite abelian groups. 

\section{\label{additivity} Superadditivity of $R_G(\psi)$}

The $G$-asymmetry given in Eq.~\eqref{1} is an ensemble frameness monotone~\cite{GourSpekkens,Gour}, i.e.~it does not increase on average under $G$-invariant operations. However, the regularized, linearized $G$-asymmetry, $A_{G}^{(reg)}$ (see Eq.(\ref{regu})), is not necessarily 
an ensemble monotone (as was shown for the case where $G=\mathbb{Z}_2$ in~\cite{GourSpekkens}). Nonetheless,
as the linearization function is monotonically increasing, $A_{G}^{(reg)}$ must be a deterministic frameness monotone. That is, $A_{G}^{(reg)}$, does not increase under deterministic $G$-invariant operations (see~\cite{GourSpekkens} for more details). As such $A_{G}^{(reg)}$ quantifies how well a quantum state can serves as a token of the missing reference frame.  

In addition,  from its definition, $A_{G}^{(reg)}$ is weakly additive; i.e.~$A_{G}^{(reg)}(\psi^{\otimes 2})=2A_{G}^{(reg)}(\psi)$.
As $R_{G}(\psi)=A_{G}^{(reg)}(\psi)$ for $G=U(1)$ and $G=\mathbb{Z}_{M}$, it follows that 
$R_{G}(\psi)$ (for $G=U(1)$ or $G=\mathbb{Z}_M$) is also a deterministic frameness monotone that is weakly additive as one would intuitively expect.
The question we address in this section is whether $R_{G}$ is also strongly additive; that is, for any two pure states $\ket{\psi}$ and $\ket{\phi}$ is it true that
\begin{equation}
R_{G}(\psi\otimes\phi)=R_{G}(\psi)+R_{G}(\phi)\;?
\label{superadditive}
\end{equation}

In the case where $G=U(1)$ Eq.~\eqref{superadditive} is true. Indeed, we have seen in the previous section that 
for $G=U(1)$, the alignment rate $R_{U(1)}(\psi)=4\pi V(\psi)$, where $V(\psi)$ is the number variance of the state $\ket{\psi}$. It was shown in~\cite{GourSpekkens} that the number variance is strongly additive as $V(\psi\otimes \phi)=V(\psi)+V(\phi)$ for any two states $\ket{\psi}$, $\ket{\phi}$. Note that $R_{U(1)}(\psi)$ is strongly additive because it is equal to the number variance.  One cannot infer the strong additivity of $R_{U(1)}(\psi)$ from its definition without the explicit calculation of the previous section.

Many open questions in the field of quantum information theory are concerned with the strong additivity of regularized quantities.  For example, it is believed by many researchers that the classical capacity of a quantum channel, $\mathcal{C}(\mathcal{N})=\lim_{n\to\infty}\frac{\chi(\mathcal{N}^{\otimes n})}{n}$ where $\chi(\mathcal{N})$ is the Holevo capacity~\cite{Holevo2, SW}, is strongly additive.     
We now show that for some groups, $R_{G}$ is not strongly additive even under the tensor product of two (distinct) \emph{pure} states. 

Suppose Alice and Bob lack a shared frame of reference associated with the group $\mathbb{Z}_M$.  Consider
two resource states (or bounded-size quantum reference frames~\cite{BRST}) $\ket{\psi}=\sum_{k=0}^{M-1} \sqrt{p_k}\ket{k}$ and $\ket{\phi}=\sum_{k=0}^{M-1} \sqrt{q_k}\ket{k}$ (these resource states serve as a token for the missing reference frame). From Theorem~\ref{thm:2} we have
$$
R_{\mathbb{Z}_M}(\ket{\psi})=-2\log r_{\mathrm{max}}\;\;\text{and}\;\;
R_{\mathbb{Z}_M}(\ket{\phi})=-2\log l_{\mathrm{max}}
$$
where
\begin{align*}
& r_\mathrm{max}=\max_{n\in\{1,...,M-1\}}\left|\sum_{m=0}^{M-1}e^{\frac{\imath2\pi nm}{M}}\,p_m\right|\\
& l_\mathrm{max}=\max_{n\in\{1,...,M-1\}}\left|\sum_{m=0}^{M-1}e^{\frac{\imath2\pi nm}{M}}\,q_m\right|\;.
\end{align*}  
We now calculate $R_{\mathbb{Z}_M}(\ket{\psi}\otimes\ket{\phi})$. Up to $Z_M$-invariant unitaries, $\ket{\psi}\otimes\ket{\phi}$ can be written as 
\begin{align}
\ket{\psi}\otimes\ket{\phi}=\sum_{k_1,k_2=0}^{M-1}\sqrt{p_{k_1}q_{k_2}}\ket{k_1}\otimes\ket{k_2}
=\sum_k\sqrt{c_k}\ket{k},
\label{79}
\end{align}
where $c_k=\sum p_{k_1}q_{k_2}$, and the sum is over all $k_1,\,k_2$ such that $k_1+k_2=k_{\mathrm{mod}M}$.  Computing the Fourier transform of the coefficients $c_k$, one obtains
\begin{align}
\omega_n&=\sum_{m=0}^{M-1}e^{\frac{\imath2\pi mn}{M}}c_m
&=\sum_{m=0}^{M-1}\sum_{k_1+k_2=m}e^{\frac{\imath2\pi n(k_1+k_2)}{M}}p_{k_1}q_{k_2}.
\label{80}
\end{align}
Noting that $\sum_{m=0}^{M-1}\sum_{k_1+k_2=m}=\sum_{k_1=0}^{M-1}\sum_{k_2=0}^{M-1}$, Eq.~\eqref{80} reduces to 
\begin{equation}
\omega_n=\sum_{k_1=0}^{M-1}e^{\frac{\imath2\pi k_1n}{M}}p_{k_1}\sum_{k_2=0}^{M-1}e^{\frac{\imath2\pi k_2n}{M}}q_{k_2}
\equiv r_nl_ne^{\imath(\theta_n+\phi_n)}\;,
\label{81}
\end{equation}
where $r_n$ and $l_n$ are the absolute values of the Fourier transforms of $\{p_k\}$ and $\{q_k\}$, respectively.
We therefore get
\begin{align}
R_{\mathbb{Z}_M}(\psi\otimes \phi)&=\max_{n\in\{1,...,M-1\}}\left(-2\log r_n-2\log l_n\right)\nonumber\\
& \geq  -2\log r_{\max}-2\log l_{\max}\nonumber\\
&= R_{\mathbb{Z}_M}(\psi)+R_{\mathbb{Z}_M}(\phi)\;.
\end{align}  
Hence, $R_{\mathbb{Z}_M}$ is not strongly additive in general. 
For the case where $M=2$,  $R_{\mathbb{Z}_M}$ is strongly additive since there is only a single $n$, namely $n=1$.  For the case where $M=3$, $R_{\mathbb{Z}_M}$ is again strongly additive as there are only two values for $n$,  which turn out to satisfy $\omega_1=\omega_2^*$ and thus $|\omega_1|=|\omega_2|$ .
However, for $M\geq 4$,  $R_{\mathbb{Z}_M}$  is super-additive as the following example shows.  

Consider the case where Alice and Bob lack a shared frame of reference for $\mathbb{Z}_4$ and Alice has the states 
\begin{align}\nonumber
\ket{\psi}&=\sqrt{\frac{13}{64}}\ket{0}+\sqrt{\frac{18}{64}}\ket{1}+\sqrt{\frac{19}{64}}\ket{2}+\sqrt{\frac{14}{64}}\ket{3}\\
\ket{\phi}&=\sqrt{\frac{7}{20}}\ket{0}+\sqrt{\frac{3}{20}}\ket{1}+\sqrt{\frac{6}{20}}\ket{2}+\sqrt{\frac{4}{20}}\ket{3}.
\label{800}
\end{align}
Computing the Fourier transforms one obtains
$
r_1=r_3=0.113,\;r_2=0
$,
and
$
l_1=l_3=0.07,\;l_2=0.3
$.
Thus, $r_{\mathrm{max}}=r_1=0.113$, and $l_{\mathrm{max}}=l_2=0.3$.  Moreover, $|\omega_1|=0.008$, $|\omega_2|=0$ and $|\omega_3|=0.008$. We therefore have $|\omega_{\mathrm{max}}|=r_1l_1$ which is smaller than $r_{\max}l_{\max}$.  Thus, $R_{\mathbb{Z}_4}(\ket{\psi}\otimes\ket{\phi})$ is strictly greater than $R_{\mathbb{Z}_4}(\ket{\psi})+R_{\mathbb{Z}_4}(\ket{\phi})$.  

In order to understand the meaning of the super-additivity of $R_{G}$, suppose that Bob holds many copies of two resource states $\ket{\psi}$ and $\ket{\phi}$. In particular, consider the state $|\psi\rangle^{\otimes N}\otimes\ket{\phi}^{\otimes N}$, where
$N\gg1$. In order to learn Alice's reference frame, Bob performs measurements on the resource states sent to him by Alice.
The super-additivity of $R_{G}$ indicates that in order for Bob the learn the most about Alice's reference frame,
he should perform a joint measurement on the full system $|\psi\rangle^{\otimes N}\otimes\ket{\phi}^{\otimes N}$, rather than two separate joint measurements, one on $|\psi\rangle^{\otimes N}$ and the other on $\ket{\phi}^{\otimes N}$.  To our knowledge, the alignment rate is the first example of a regularized quantity that is not strongly additive even on distinct pure states.

\section{\label{sec:4}Conclusion}

To summarize, we have derived an information theoretic, operational interpretation for the $G$-asymmetry for the case of a phase reference, and a reference frame associated with a discrete cyclic group of $M$ elements, $\mathbb{Z}_M$.   In particular, we have shown that the alignment rate, $R_G(\psi)$, in a phase alignment protocol associated with $G=U(1)$, and an alignment protocol associated with a finite cyclic group of order $M$, $\mathbb{Z}_M$, is equal to the regularized, linearized $G$-asymmetry, $A^{\mathrm{reg}}_G(\psi)$.  As the $G$-asymmetry is equal to the Holevo bound, our result implies that for reference frames associated with $G=U(1)$ and $G=\mathbb{Z}_M$, the linearized Holevo bound is asymptotically achievable.  We are willing to conjecture that $R_{G}=A_{G}^{(reg)}$ for all finite or compact Lie groups.  

The additivity of $R_G(\psi)$ was discussed in Sec.~\ref{additivity}, where it was shown that $R_G(\psi)$ is both weakly and strongly additive for $G=U(1)$, $G=\mathbb{Z}_2$, and $G=\mathbb{Z}_3$, but only weakly additive for finite cyclic groups, $\mathbb{Z}_M$, for $M\geq4$.  For the latter, we proved that $R_{G}$ is super-additive. In the case of finite groups however, there exists a resource state, denoted here by $|+\rangle$, which \emph{completely} elevates the restrictions that follow from the lack of a shared reference frame. It is therefore the ultimate resource Bob can hold. 
It turns out that up to $\mathbb{Z}_M$-invariant unitary operations, $|\psi\rangle^{\otimes N}$ approaches the ultimate resource state $|+\rangle$ in the limit $N\to\infty$. Hence, it follows from Eq.(\ref{57}) for example, that the super-additivity of $R_{\mathbb{Z}_M}(\psi)$ indicates only an exponentially small gain of reference-frame information in the performance of a joint measurement on the full system, 
$|\psi\rangle^{\otimes N}\otimes\ket{\phi}^{\otimes N}$, rather than two separate joint measurements on $|\psi\rangle^{\otimes N}$ and $\ket{\phi}^{\otimes N}$. It is therefore 
left open if there are non-finite compact Lie groups for which $R_{G}$ is super-additive.
Such examples, if exist, will have a more significant gain (i.e. not exponentially small) of reference-frame information in the performance of a joint measurement on the full system, 
$|\psi\rangle^{\otimes N}\otimes\ket{\phi}^{\otimes N}$, rather than two separate joint measurements on $|\psi\rangle^{\otimes N}$ and $\ket{\phi}^{\otimes N}$.

\section{Acknowledgements}
We are grateful to Fernando Brand\~{a}o, Giulio Chiribella, Iman Marvian, Aram W.~Harrow, and Robert W.~Spekkens for helpful discussions on this work.  We gratefully acknowledge support from NSERC and USARO. 

\bibliography{ratebiblioJune21}
\end{document}